\documentclass[english,nolineno]{socg-lipics-v2019}
\newcommand{\PAPER}[1]{#1}
\newcommand{\SOCG}[1]{}

\usepackage[utf8]{inputenc}

\usepackage{microtype,hyperref}
   \hypersetup{%
      breaklinks,%
      ocgcolorlinks, colorlinks=true,%
      urlcolor=[rgb]{0.25,0.0,0.0},%
      linkcolor=[rgb]{0.5,0.0,0.0},%
      citecolor=[rgb]{0,0.2,0.445},%
      filecolor=[rgb]{0,0,0.4},
      anchorcolor=[rgb]={0.0,0.1,0.2}%
   }

\usepackage{amsmath,amssymb,amsthm,amsfonts,dsfont,mathtools}
\usepackage{xspace}
\usepackage{algorithm}
\usepackage[noend]{algpseudocode}

\algnewcommand{\Inputs}[1]{%
  \State \textbf{Inputs:}
  \Statex \hspace*{\algorithmicindent}\parbox[t]{.8\linewidth}{\raggedright #1}
}
\algnewcommand{\Initialize}[1]{%
  \State \textbf{Initialize:}
  \Statex \hspace*{\algorithmicindent}\parbox[t]{.8\linewidth}{\raggedright #1}
}
\algnewcommand{\TurnOne}[1]{%
  \State \textbf{Timestep 1:}
  \Statex \hspace*{\algorithmicindent}\parbox[t]{.8\linewidth}{\raggedright #1}
}

\newcommand{\eps}{\varepsilon}

\newcommand{\ignore}[1]{}

\newcommand{\OPT}{\textrm{OPT}}

\makeatletter
\DeclareRobustCommand\onedot{\futurelet\@let@token\@onedot}
\def\@onedot{\ifx\@let@token.\else.\null\fi\xspace}

\makeatother

\EventEditors{Sergio Cabello and Danny Z. Chen}
\EventNoEds{2}
\EventLongTitle{36th International Symposium on Computational Geometry (SoCG 2020)}
\EventShortTitle{SoCG 2020}
\EventAcronym{SoCG}
\EventYear{2020}
\EventDate{June 23--26, 2020}
\EventLocation{Z\"{u}rich, Switzerland}
\EventLogo{socg-logo}
\SeriesVolume{164}

\title{Faster Approximation Algorithms for Geometric Set Cover}

\author{Timothy M. Chan}{Department of Computer Science, University of Illinois at Urbana-Champaign, USA}{tmc@illinois.edu}{https://orcid.org/0000-0002-8093-0675}{
Supported in part by NSF Grant CCF-1814026.}

\author{Qizheng He}{Department of Computer Science, University of Illinois at Urbana-Champaign, USA}{qizheng6@illinois.edu}{}{}

\authorrunning{T.\,M. Chan and Q. He}
\Copyright{Timothy M. Chan and Qizheng He}



\hideLIPIcs  

\begin{document}
\maketitle

\sloppy
\begin{abstract}
We improve the running times of $O(1)$-approximation algorithms for the set cover problem in geometric settings, specifically, covering points by disks in the plane, or covering points by halfspaces in three dimensions. In the unweighted case, Agarwal and Pan [SoCG 2014] gave a randomized $O(n\log^4 n)$-time, $O(1)$-approximation algorithm, by using variants of the multiplicative weight update (MWU) method combined with geometric data structures. We simplify the data structure requirement in one of their methods and obtain
a \emph{deterministic} $O(n\log^3 n\log\log n)$-time algorithm.  With further new ideas, we obtain a still faster randomized $O(n\log n(\log\log n)^{O(1)})$-time algorithm.

For the weighted problem, we also give a randomized $O(n\log^4n\log\log n)$-time, $O(1)$-approximation algorithm, by simple modifications to the MWU method and the quasi-uniform sampling technique.


\end{abstract}

\section{Introduction}\label{sec:intro}

\subparagraph*{Unweighted geometric set cover.}
In this paper we study one of the most fundamental classes of geometric optimization problems: \emph{geometric set cover}.  Given a set $X$ of $O(n)$ points and a set $S$ of $O(n)$ geometric objects, find the smallest subset of objects from $S$ to cover all points in $X$.  In the dual set system, the problem corresponds to \emph{geometric hitting set} (finding the smallest number of points from $X$ that hit all objects in $S$).

This class of problems has been \emph{extensively} investigated in the computational geometry literature.  Since they are NP-hard in most scenarios, attention is turned towards approximation algorithms.  Different types of objects give rise to different results.
Typically, approximation algorithms fall into the following categories:

\begin{enumerate}
    \item Simple heuristics, e.g., greedy algorithms.
    \item Approaches based on solving the linear programming (LP) relaxation (i.e., fractional set cover) and rounding the LP solution.
    \item Polynomial-time approximation schemes (PTASs), e.g., via local search, shifted grids/quadtrees (sometimes with dynamic programming), or separator-based divide-and-conquer.
\end{enumerate}

Generally, greedy algorithms achieve only logarithmic approximation factors (there are some easy cases where they give $O(1)$ approximation factors, e.g., hitting set for fat objects such as disks/balls in the ``continuous'' setting with $X=\mathbb{R}^d$~\cite{EfratKNS00}).  The LP-based approaches give better approximation factors in many cases, e.g., $O(1)$ approximation for set cover and hitting set for disks in 2D and halfspaces in 3D, set cover for objects in 2D with linear ``union complexity'', and hitting set for pseudodisks in 2D \cite{bronnimann1995almost,clarkson2007improved,AronovES10,Varadarajan09,PyrgaR08}.
Subsequently, local-search PTASs have been found by Mustafa and Ray~\cite{mustafa2009ptas} in some cases, including set cover and hitting set for disks in 2D and halfspaces in 3D (earlier, PTASs were known for hitting set only in the continuous setting for unit disks/balls~\cite{HochbaumM85}, and for arbitrary disks/balls and fat objects~\cite{Chan03}).

Historically, the focus has been on obtaining good approximation factors.  Here, we are interested in obtaining approximation algorithms with good---ideally, near linear---running time.  Concerning the efficiency of known approximation algorithms:

\begin{enumerate}
    \item
Certain simple heuristics can lead to fast $O(1)$-approximation algorithms in some easy cases (e.g., continuous hitting set for unit disks or unit balls by using grids), but generally, even simple greedy algorithms may be difficult to implement in near linear time (as they may require nontrivial dynamic geometric data structures).
\smallskip
\item
LP-based approaches initially may not seem to be the most efficient, because of the need to solve an LP\@.  However, a general-purpose LP solver can be avoided.  The set-cover LP can alternatively be solved (approximately) by the \emph{multiplicative weight update} (MWU) method. In the computational geometry literature, the technique has been called \emph{iterative reweighting}, and its use in geometric set cover was explored by Br\"onnimann and Goodrich~\cite{bronnimann1995almost} (one specific application appeared in an earlier work by Clarkson~\cite{clarkson1993algorithms}), although the technique was known even earlier outside of geometry.  On the other hand, the LP-rounding part corresponds to the well-known geometric problem of constructing \emph{$\eps$-nets}, for which efficient algorithms are known~\cite{chan2016optimal,Mustafa19}.
\smallskip
\item
PTAS approaches generally have large polynomial running time, even when specialized to specific approximation factors.  For example, see~\cite{BusGMR17} for efforts in improving the degree of the polynomial.
\end{enumerate}

In this paper we design faster approximation algorithms for geometric set cover via the LP/MWU-based approaches.  There has been a series of work on speeding up MWU methods for covering or packing LPs (e.g., see \cite{chekuri2018randomized,KoufogiannakisY14,Young01}).
In geometric settings, we would like more efficient algorithms (as generating the entire LP explicitly would already require quadratic time), by somehow exploiting geometric data structures.  The main previous work was by Agarwal and Pan \cite{agarwal2014near} from SoCG 2014, who showed how to compute an $O(1)$-approximation for set cover for 2D disks or 3D halfspaces, in $O(n\log^4n)$ randomized time.

Agarwal and Pan actually proposed two MWU-based algorithms:
The first is a simple variant of the standard MWU algorithm of Br\"onnimann and Goodrich, which proceeds in logarithmically many rounds.
The second  views the problem as a 2-player zero-sum game, works quite differently (with weight updates to both points and objects), and uses randomization; the analysis is more complicated.
Because the first algorithm requires stronger data structures---notably, for approximate weighted range counting with dynamic changes to the weights---Agarwal and Pan chose to implement their second algorithm instead, to get their $O(n\log^4n)$ result for 3D halfspaces.

\subparagraph*{New results.}
In this paper we give:

\begin{itemize}
    \item a \emph{deterministic} near-linear $O(1)$-approximation algorithm for set cover for 3D halfspaces.  Its running time is $O(n\log^3 n\log\log n)$, which besides eliminating randomization is also a little faster than Agarwal and Pan's;
    \medskip
    \item a still faster randomized near-linear $O(1)$-approximation algorithm for set cover for 3D halfspaces.  Its running time is $O(n\log n\log^{O(1)}\log n)$, which is essentially optimal\footnote{Just deciding whether a solution exists requires $\Omega(n\log n)$ time in the algebraic decision-tree model, even for 1D intervals.}
    ignoring minor $\log\log n$ factors.
\end{itemize}

Although generally shaving logarithmic factors may not be the most important endeavor, the problem is fundamental enough that we feel it worthwhile to find the most efficient algorithm possible.

Our approach interestingly is to go back to Agarwal and Pan's first MWU algorithm.  We show that with one simple modification, the data structure requirement can actually be relaxed: namely, for the approximate counting structure, there is no need for weights, and the only update operation is insertion.  By standard techniques, insertion-only data structures reduce to static data structures.
This simple idea immediately yields our deterministic result.
(Before, Bus et al.~\cite{bus2018practical} also aimed to find variants of Agarwal and Pan's first algorithm with simpler data structures, but they did not achieve improved theoretical time bounds.)  Our best randomized result requires a more sophisticated combination of several additional ideas.  In particular, we incorporate random sampling in the MWU algorithm, and extensively use \emph{shallow cuttings}, in both primal and dual space.

We have stated our results for set cover for 3D halfspaces.  This case is arguably the most central. It is equivalent to hitting set for 3D halfspaces, by duality, and also includes set cover and hitting set for 2D disks as special cases, by the standard lifting transformation.  The case of 3D dominance ranges is another special case, by a known transformation \cite{chan2011orthogonal,PachT11} (although for the dominance case, word-RAM techniques can speed up the algorithms further).  The ideas here are likely useful also in the other cases considered in Agarwal and Pan's paper (e.g., hitting set for rectangles, set cover for fat triangles, etc.), but in the interest of keeping the paper focused, we will not discuss these implications.

\subparagraph*{Weighted geometric set cover.}
Finally, we consider the weighted version of set cover: assuming that each object is given a weight, we now want a subset of the objects of $R$ with the minimum total weight that covers all points in $X$.
The weighted problem has also received considerable attention:
Varadarajan~\cite{varadarajan2010weighted} and
Chan et al.~\cite{chan2012weighted} used the LP-based approach to obtain $O(1)$-approximation algorithms for weighted set cover for 3D halfspaces (or for objects in 2D with linear union complexity); the difficult part is in constructing $\eps$-nets with small weights, which they solved by the
\emph{quasi-random sampling} technique.
Later, Mustafa, Raman, and Ray~\cite{MustafaRR15} discovered
a quasi-PTAS for 3D halfspaces by using geometric separators; the running time is very high $(n^{\log^{O(1)}n})$.

Very recently, Chekuri, Har-Peled, and Quanrud~\cite{fasterlp} described new randomized MWU methods which can efficiently solve the LP corresponding to various generalizations of geometric set cover, by using appropriate geometric data structures.
In particular, for weighted set cover for 3D halfspaces, they obtained a randomized $O(n\log^{O(1)}n)$-time algorithm to solve the LP but with an unspecified number of logarithmic factors.
They did not address the LP-rounding part, i.e., construction of an $\eps$-net of small weight---a direct implementation of the quasi-uniform sampling technique would not lead to a near-linear time bound.

We observe that a simple direct modification of the standard MWU algorithm of Br\"onnimann and Goodrich, or Agarwal and Pan's first algorithm, can also solve the LP for weighted geometric set cover, with arguably simpler data structures than Chekuri et al.'s.  Secondly, we observe that an $\eps$-net of small weight can be constructed in near-linear time, by using quasi-uniform sampling more carefully.
This leads to a randomized $O(n\log^4n\log\log n)$-time, $O(1)$-approximation algorithm for weighted set cover for 3D halfspaces (and thus for 2D disks).






\section{Preliminaries}

Let $X$ be a set of points and $S$ be a set of objects.
For a point $p$, its \emph{depth} in $S$ refers to the number of objects in $S$ containing $p$.  A point $p$ is said to be \emph{$\eps$-light} in $S$, if it has depth $\leq\eps |S|$ in $S$; otherwise it is \emph{$\eps$-heavy}.
A subset of objects $T\subseteq S$ is an \emph{$\eps$-net} of $S$ if $T$ covers all points that are $\eps$-heavy in $S$.

It is known that there exists an $\eps$-net with size $O(\frac{1}{\eps})$ for any set of halfspaces in 3D or disks in 2D~\cite{matouvsek1990net} (or more generally for objects in the plane with linear union complexity~\cite{clarkson2007improved}).


\subsection{The Basic MWU Algorithm}
We first review the standard multiplicative weight\footnote{
In our algorithm description, we prefer to use the term ``multiplicity'' instead of ``weight'', to avoid confusion with the weighted set cover problem later.
} update (MWU) algorithm for geometric set cover, as described by Br\"onnimann and Goodrich~\cite{bronnimann1995almost} (which generalizes an earlier algorithm by Clarkson~\cite{clarkson1993algorithms}, and is also well known outside of computational geometry).

Let $X$ be the set of input points and $S$ be the set of input objects, with $n=|X|+|S|$. Let $\OPT$ denote the size of the minimum set cover.
We assume that a value $t=\Theta(\OPT)$ is known; this assumption will be removed later by a binary search for $t$.
In the following pseudocode, we work with a multiset $\hat{S}$; in measuring size or counting depth, we include multiplicities (e.g., $|\hat{S}|$ is the sum of the multiplicities of all its elements).


\begin{algorithm}[H]
\begin{algorithmic}[1]
\State Guess a value $t\in [\OPT,2\,\OPT]$ and set $\eps=\frac{1}{2t}$.
\State Define a multiset $\hat{S}$ where each object $i$ in $S$ initially has multiplicity $m_i=1$.
\While {we can find a point $p\in X$ which is $\eps$-light in $\hat{S}$}
    \For {each object $i$ containing $p$}
    \Comment{call lines 4--5 a \emph{multiplicity-doubling step}}
        \State Double its multiplicity $m_i$.
    \EndFor
\EndWhile
\State Return an $\eps$-net of the multiset $\hat{S}$.
\end{algorithmic}
\end{algorithm}
Since at the end all points in $X$ are $\eps$-heavy in $\hat{S}$, the returned subset is a valid set cover of $X$.  For halfspaces in 3D or disks in 2D, its size is $O(\frac{1}{\eps})=O(t)=O(\OPT)$.

A standard analysis shows that the algorithm always terminates after
$O(t\log\frac{n}{t})$ multiplicity-doubling steps.  We include a quick proof:
Each multiplicity-doubling step increases $|\hat{S}|$ by a factor of at most $1+\eps$, due to the $\eps$-lightness of~$p$.  Thus, after $z$ doubling steps,
$|\hat{S}|\le n(1+\eps)^z \le n e^{\eps z} = ne^{z/(2t)}$.
On the other hand, consider a set cover $T^*$ of size $t$. In each multiplicity-doubling step,
at least one of the objects in $T^*$ has its multiplicity doubled.
So, after $z$ multiplicity-doubling steps, the total multiplicity in $T^*$
is at least $t2^{z/t}$.
We conclude that $t2^{z/t}\le |\hat{S}|\le ne^{z/(2t)}$, implying that $z=O(t\log\frac{n}{t})$.

\subsection{Agarwal and Pan's (First) MWU Algorithm}

Next, we review Agarwal and Pan's first variant of the
MWU algorithm~\cite{agarwal2014near}.
One issue in implementing the original algorithm lies in the test in line~3: searching for one light point by scanning all points in $X$ from scratch every time seems inefficient.  In Agarwal and Pan's refined approach, we proceed in a small number of rounds, where in each round, we examine the points in $X$ in a fixed order and test for lightness in that order.

\begin{algorithm}[H]
\begin{algorithmic}[1]
\State Guess a value $t\in [\OPT,2\,\OPT]$ and set $\eps=\frac{1}{2t}$.
\State Define a multiset $\hat{S}$ where each object $i$ in $S$ initially has multiplicity $m_i=1$.
\Loop \Comment{call this the start of a new \emph{round}}
    \For {each point $p\in X$ in any fixed order}
        \While {$p$ is $\eps$-light in $\hat{S}$}
                \For {each object $i$ containing $p$} \Comment{call lines 6--7 a \emph{multiplicity-doubling step}}
                    \State Double its multiplicity $m_i$.
                \EndFor
                \If {the number of multiplicity-doubling steps in this round exceeds $t$}
                    \State Go to line~3 and start a new round.
                \EndIf
        \EndWhile
    \EndFor
    \State Terminate and return an $\frac\eps2$-net of the multiset $\hat{S}$.
\EndLoop
\end{algorithmic}
\end{algorithm}

To justify correctness, observe that since each round performs at most $t$ multiplicity-doubling steps, $|\hat{S}|$ increases by a factor of at most
$(1+\eps)^t\le e^{\eps t}\le e^{1/2} < 2$.  Thus, a point $p$ that is checked to be $\eps$-heavy in $\hat{S}$ at any moment during the round will remain
$\frac\eps2$-heavy in $\hat{S}$ at the end of the round.

Since all but the last round performs $t$ multiplicity-doubling steps and we have already shown that the total number of such steps is $O(t\log\frac{n}{t})$, the number of rounds is
$O(\log\frac{n}{t})$.


\section{``New'' MWU Algorithm}

Agarwal and Pan's algorithm still requires an efficient data structure to test whether a given point is light, and the data structure needs to support dynamic changes to the multiplicities.
We propose a new variant that requires simpler data structures.

Our new algorithm is almost identical to Agarwal and Pan's, but with just one very simple change!  Namely, after line~3, at the beginning of each round, we add the following line, to readjust all multiplicities:

\noindent\hrulefill

\noindent\ \ {\small 3.5:}\ \ \ \  for each object $i$,
reset its multiplicity $m_i\leftarrow \lceil m_i \frac{10n}{|\hat{S}|}\rceil$.

\noindent\hrulefill

\smallskip
To analyze the new algorithm, consider modifying the multiplicity $m_i$ instead to $\lceil m_i \frac{10n}{|\hat{S}|}\rceil\cdot \frac{|\hat{S}|}{10n}$.
The algorithm behaves identically (since the multiplicities are identical except for a common rescaling factor), but is more convenient to analyze.  In this version, multiplicities are nondecreasing over time (though they may be non-integers).  After the modified line~3.5, the new $|\hat{S}|$ is at most $\sum_i \left(m_i \frac{10n}{|\hat{S}|} + 1\right)\cdot \frac{|\hat{S}|}{10n} \le 11n\cdot \frac{|\hat{S}|}{10n} = 1.1|\hat{S}|$.
If the algorithm makes $z$ multiplicity-doubling steps, then it performs line~3.5 at most $z/t$ times and we now have $|\hat{S}|\le n (1+\eps)^z\cdot 1.1^{z/t}\le ne^{z/(2t)}\cdot 1.1^{z/t}$.  This is still sufficient to imply that $z=O(t\log\frac nt)$, and so the number of rounds remains
$O(\log\frac nt)$.

Now, let's go back to line~3.5 as written.  The advantage of this multiplicity readjustment step is that it decreases $|\hat{S}|$ to $\sum_i \left(m_i \frac{10n}{|\hat{S}|} + 1\right) = O(n)$.
At the end of the round, $|\hat{S}|$ increases by a factor of at most $(1+\eps)^t < 2$ and so remains $O(n)$.
Thus, in line~7, instead of doubling the multiplicity of an object, we can just repeatedly increment the multiplicity (i.e., insert one copy of an object) to reach the desired value.  The total number of increments per round is $O(n)$.

\newcommand{\Tinsert}{T_{\textrm{insert}}}
\newcommand{\Tquery}{T_{\textrm{count}}}
\newcommand{\Treport}{T_{\textrm{report}}}
\newcommand{\Tprep}{T_{\textrm{prep}}}
\newcommand{\Tcount}{T_{\textrm{count}}}
\newcommand{\Tnet}{T_{\textrm{net}}}
\newcommand{\REPORT}{{\sc Report}}
\newcommand{\COUNT}{{\sc Approx-Count-Decision}}

Note that in testing for $\eps$-lightness in line~5, a constant-factor approximation of the depth is sufficient, with appropriate adjustments of constants in the algorithm.  Also, although the algorithm as described may test the same point $p$ for lightness several times in a round, this can be easily avoided: we just keep track of the increase $D$ in the depth of the current point~$p$; the new depth of $p$ can be 2-approximated by the maximum of the old depth and $D$.

To summarize, an efficient implementation of each round of the new algorithm requires solving the following geometric data structure problems (\REPORT\ for line~6, and \COUNT\ for line~5):

\begin{description}
    \item[Problem \REPORT:] Design a data structure to store a static set $S$ of size $O(n)$ so that given a query point $p\in X$, we can report all objects in $S$ containing the query point $p$.  Here, the output size of a query is guaranteed to be at most $O(k)$ where $k := \frac nt$ (since $\eps$-lightness of~$p$ implies that its depth is at most $\eps |\hat{S}|=\Theta(\frac nt)$ even including multiplicities).
    \medskip
    \item[Problem \COUNT:] Design a data structure to store a multiset $\hat{S}$ of size $O(n)$ so that given a query point $p\in X$, we can either declare that the number of objects in $\hat{S}$ containing $p$ is less than a fixed threshold value $k$, or that the number is more than $\frac kc$, for some constant $c>1$.  Here, the threshold again is $k := \frac nt$ (since $\eps|\hat{S}|=\Theta(\frac nt)$).
    The data structure should support the following type of updates: insert one copy of an object to $\hat{S}$.  (Deletions are not required.)
    Each point in $X$ is queried once.
\end{description}

\noindent
To bound the cost of the algorithm:

\begin{itemize}
\item Let $\Treport$ denote the total time for $O(t)$ queries in Problem \REPORT.
\item Let $\Tcount$ denote the total time for $O(n)$ queries and $O(n)$ insertions in Problem \COUNT.  (Note that the initialization of $\hat{S}$ at the beginning of the round can be done by $O(n)$ insertions.)
\item Let $\Tnet$ denote the time for computing an $\eps$-net of size $O(\frac1\eps)$ for a given multiset $\hat{S}$ of size $O(n)$.
\end{itemize}

\noindent The total running time over all $O(\log\frac nt)$ rounds is
\begin{equation}\label{eqn:runtime}
O((\Treport+\Tcount)\log\tfrac nt \,+\, \Tnet).
\end{equation}

\section{Implementations}

In this section,
we describe specific implementations of our MWU algorithm when the objects are halfspaces in 3D (which include disks in 2D as a special case by the standard lifting transformation).  We first consider deterministic algorithms.

\subsection{Deterministic Version}\label{sec:det}

\subparagraph*{Shallow cuttings.}
We begin by reviewing an important tool that we will use several times later.
For a set of $n$ planes in $\mathbb{R}^3$, a \emph{$k$-shallow $\eps$-cutting}
is a collection of interior-disjoint polyhedral cells, such that each cell intersects at most $\eps n$ planes, and the union of the cells cover all points of level at most $k$ (the \emph{level} of a point refers to the number of planes below it).  The list of all planes intersecting a cell $\Delta$ is called the \emph{conflict list} of $\Delta$.
Matou\v{s}ek \cite{matousek1992reporting} proved the existence of a $k$-shallow $(\frac{ck}{n})$-cutting with $O(\frac nk)$ cells for any
constant~$c$.  Chan and Tsakalidis \cite{chan2016optimal} gave an $O(n\log \frac nk)$-time deterministic algorithm to construct such a cutting, along with all its conflict lists (an earlier randomized algorithm was given by Ramos~\cite{ramos1999range}).
If $c$ is sufficiently large, the cells may be made ``downward'', i.e., they all contain $(0,0,-\infty)$.

\subparagraph*{Constructing $\eps$-nets.}
The best known deterministic algorithm for constructing $\eps$-nets for 3D halfspaces is by Chan and Tsakalidis~\cite{chan2016optimal} and runs in $\Tnet = O(n\log\frac1\eps)=O(n\log n)$ time.

The result follows directly from their shallow cutting algorithm (using a simple argument of Matou\v sek~\cite{matousek1992reporting}):  Without loss of generality, assume that all halfspaces are upper halfspaces, so depth corresponds to level with respect to the bounding planes (we can compute a net for lower halfspaces separately and take the union, with readjustment of $\eps$ by a factor of~2).
We construct an $(\eps n)$-shallow $\frac\eps2$-cutting with $O(\frac1\eps)$ cells, and for each cell, add a plane completely below
the cell (if it exists) to the net.
To see correctness, for a point $p$ with level $\eps n$, consider the cell $\Delta$ containing $p$; at least $\eps n - \frac{\eps n}{2} >0$ planes are completely below $\Delta$, and so the net contains at least one plane below $p$.

\subparagraph*{Solving Problem \REPORT.}
This problem corresponds to 3D \emph{halfspace range reporting} in dual space,
and by known data structures~\cite{Chan00,
afshani2009optimal,chan2016optimal}, the total time to answer $O(t)$ queries is $\Treport =O(t\cdot (\log n + k)) = O(t\log n + n)$, assuming an initial preprocessing of $O(n\log n)$ time (which is done only once).

This result also follows directly from shallow cuttings (since space is not our concern, the solution is much simplified):
Without loss of generality, assume that all halfspaces are upper halfspaces.  We construct a $k$-shallow $O(\frac kn)$-cutting with $O(\frac nk)$ downward cells.  Given a query point $p\in X$, we find the cell containing $p$, which can be done in $O(\log n)$ time by planar point location; we then do a linear search over its conflict list, which has size $O(k)$.

Note that the point location operations can be actually be done during preprocessing in $O(n\log n)$ time since $X$ is known in advance.  This lowers the time bound for $O(t)$ queries to $\Treport=O(tk)=O(n)$.

\subparagraph*{Solving Problem \COUNT.}
This problem corresponds to the decision version of 3D \emph{halfspace approximate range counting} in dual space, and several deterministic and randomized data structures have already been given in the static case~\cite{AfshaniC09,afshani2010general}, achieving $O(\log n)$ query time and $O(n\log n)$ preprocessing time.

This result also follows directly from shallow cuttings:
Without loss of generality, assume that all halfspaces are upper halfspaces.
We construct a $\frac{n}{b^i}$-shallow $O(\frac{1}{b^i})$-cutting with $O(b^i)$ downward cells for every $i=1,\ldots,\log_b n$ for some constant~$b$.  Chan and Tsakalidis's algorithm can actually construct all $O(\log n)$ such cuttings in $O(n\log n)$ total time.
With these cuttings, we can compute an $O(1)$-approximation to the depth/level of a query point $p$ by simply finding the largest $i$ such that $p$ is contained in a cell of the $\frac{n}{b^i}$-shallow cutting (the level of $p$ would then be $O(\frac{n}{b^i})$ and at least $\frac{n}{b^{i+1}}$).  In Chan and Tsakalidis's construction,
each cell in one cutting intersects $O(1)$ cells in the next cutting,
and so we can locate the cells containing $p$ in $O(1)$ time per $i$, for a total of $O(\log n)$ time.

To solve Problem~\COUNT, we still need to support insertion.
Although the approximate decision problem is not decomposable, the above solution solves the approximate counting problem, which is decomposable,
so we can apply the standard \emph{logarithmic method}~\cite{bentley1980decomposable} to transform the static data structure into a semi-dynamic, insertion-only data structure.
 The transformation causes a logarithmic factor increase, yielding in our case $O(\log^2n)$ query time and $O(\log^2n)$ insertion time.  Thus, the total time for $O(n)$ queries and insertions is $\Tcount=O(n\log^2 n)$.

\subparagraph*{Conclusion.}
By~(\ref{eqn:runtime}), the complete algorithm has running time
$O((\Treport + \Tcount)\log\frac nt + \Tnet) = O((n + n\log^2 n)\log\frac nt + n\log n) = O(n\log^3 n)$.

One final issue remains: we have assumed that a value $t\in [\OPT,2\,\OPT]$ is given.  In general, either the algorithm produces a solution of size $O(t)$, or (if it fails to complete within $O(\log\frac nt)$ rounds) the algorithm may conclude that $\OPT>t$.
We can thus find an $O(1)$-approximation to $\OPT$ by a binary search over $t$ among the $O(\log n)$ possible powers of~2, with $O(\log\log n)$ calls to the algorithm.  The final time bound is $O(n\log^3 n\log\log n)$.

\begin{theorem}
Given $O(n)$ points and $O(n)$ halfspaces in $\mathbb{R}^3$, we can find a subset of halfspaces covering all points, of size within $O(1)$ factor of the minimum, in deterministic
$O(n\log^3n\log\log n)$ time.
\end{theorem}

\subsection{Randomized Version 1}\label{sec:rand}

We now describe a better solution to Problem~\COUNT, by using randomization and the fact that all query points (namely, $X$) are given in advance.

\subparagraph*{Reducing the number of insertions in Problem \COUNT.}
In solving Problem~\COUNT, one simple way to speed up insertions is to work with a random sample $R$ of $\hat{S}$. When we insert an object to $\hat{S}$, we independently decide to insert it to the sample $R$ with probability $\rho := \frac{c_0\log n}{k}$, or ignore it with probability $1-\rho$, for a sufficiently large constant $c_0$.  (Different copies of an object are treated as different objects here.)
It suffices to solve the problem for the sample $R$ with the new threshold around $\rho k$.

To justify correctness, consider a fixed query point $p\in X$.  Let $x_1,x_2,\ldots$ be the sequence of objects in $\hat{S}$ that contain $p$, in the order in which they are inserted (extend the sequence arbitrarily to make its length greater than $k$).  Let $y_i=1$ if object $x_i$ is chosen to be in the sample $R$, or 0 otherwise.
Note that the $x_i$'s may not be independent (since the object we insert could depend on random choices made before); however, the $y_i$'s are independent.  By the Chernoff bound,
$\sum_{i=1}^{k/c} y_i \le \frac{(1+\delta)\rho k}{c}$ and $\sum_{i=1}^{k} y_i \ge (1-\delta)\rho k$ with probability $1-e^{-\Omega(\rho k)}=1-n^{-\Omega(c_0)}$ for any fixed constant $\delta>0$.
Thus, with high probability, at any time, if the number of objects in $R$ containing $p$ is more than $\frac{(1+\delta)\rho k}{c}$, then the number of objects in $\hat{S}$ containing $p$ is more than $\frac{k}{c}$; if the former number is less than $(1-\delta)\rho k$, then the later number is less than $k$.  Since there are $O(n)$ possible query points, all queries are correct with high probability.

By this strategy, the number of insertions is reduced to
$O(\rho n) = O(\frac{n}{k}\log n)=O(t\log n)$ with high probability.

\subparagraph*{Preprocessing step.}
Next we use a known preprocessing step to
ensure that each object contains at most $\frac nt$ points, in the case of 3D halfspaces.
This subproblem was addressed in Agarwal and Pan's paper~\cite{agarwal2014near} (where it was called ``(P5)''---curiously, they used it to implement their second algorithm but not their first MWU-based algorithm.)  We state a better running time:




\begin{lemma}\label{lemma:preprocessing}
In $O(n\log t)$ time, we can find a subset $T_0\subseteq S$ of $O(t)$ halfspaces, such that after removing all points in $X$ covered by $T_0$, each halfspace of $S$ contains at most
$\frac{n}{t}$ points.
\end{lemma}

\begin{proof}
We may assume that all halfspaces are upper halfspaces.  We work in dual space, where $S$ is now a set of points and $X$ is a set of planes.  The goal is to find a subset $T_0\subseteq S$ of $O(t)$ points such that after removing all planes of $X$ that are below some points of $T_0$,
each point of $S$ has depth/level at most $\frac nt$.

We proceed in rounds.  Let $b$ be a constant.
In the $i$-th round, assume that all points of $S$ have level $\leq \frac{n}{b^i}$. Compute a $\frac{n}{b^i}$-shallow $\frac{1}{b^{i+1}}$-cutting with $O(b^i)$ cells.
In each cell, add an arbitrary point of $S$ (if exists) to the set $T_0$.
In total $O(b^{i})$ points are added. Remove all planes that are below these added points from $X$.

Consider a point $p$ of $S$. Let $\Delta$ be the cell containing $p$, and let $q$ be the point in $\Delta$ that was added to $T_0$.  Any plane that is below $p$ but not removed (and thus above $q$) must intersect $\Delta$, so there can be at most $\frac{n}{b^{i+1}}$ such planes.  Thus, after the round, the level of $p$ is at most $\frac{n}{b^{i+1}}$. We terminate when $b^i$ reaches $t$. The total size of $T_0$ is $O(\sum_{i=1}^{\log_b t}b^i)=O(t)$.

Naively computing each shallow cutting from scratch by Chan and Tsakalidis's algorithm would require $O(n\log n\cdot\log n)=O(n\log^2 n)$  total time.  But Chan and Tsakalidis's approach can compute multiple shallow cuttings more quickly: given a $\frac{n}{b^i}$-shallow cutting along with its conflict lists,  we can compute the next $\frac{n}{b^{i+1}}$-shallow cutting along with its conflict lists in $O(n + b^i\log b^i)$ time.
However, in our application, before computing the next cutting,
we also remove some of the input planes.  Fortunately, this type of scenario has been examined in a recent paper by Chan~\cite{chan2019dynamic}, who shows that the approach still works, provided that the next cutting is relaxed to cover only points covered by the previous cutting (see Lemma~8 in his paper); this is sufficient in our application.
In our application, we also need to locate the cell containing each point of $S$.  This can still be done in $O(n)$ time given the locations in the previous cutting.  Thus, the total time is
$O(\sum_{i=1}^{\log_b t} (n + b^i\log b^i))=O(n\log t)$.
\end{proof}

At the end, we add $T_0$ back to the solution, which still has $O(t)$ total size.

\subparagraph*{Solving Problem \COUNT.}
We now propose a very simple approach to solve Problem \COUNT: just explicitly maintain the depth of all points in $X$.
Each query then trivially takes $O(1)$ time.
When inserting an object, we find all points contained in the object and increment their depths.

Due to the above preprocessing step, the number of points contained in the object is $O(\frac nt)$.
For the case of 3D halfspaces, we can find these points by halfspace range reporting; as explained before for Problem \REPORT, this can be done in $O(\frac nt)$ time by using shallow cuttings, after an initial preprocessing in $O(n\log n)$ time.
Thus, each insertion takes $O(\frac nt)$ time.
Since the number of insertions has been reduced to $O(t\log n)$ by sampling, the total time for Problem \COUNT\ is $\Tcount=O((t\log n) \cdot \frac nt) = O(n\log n)$.

\subparagraph*{Conclusion.}
By~(\ref{eqn:runtime}), the complete randomized algorithm has running time
$O((\Treport + \Tcount)\log\frac nt + \Tnet) = O((n + n\log n)\log\frac nt + n\log n) = O(n\log n\log\frac nt) = O(n\log^2 n)$ (even including the $O(n\log n)$-time preprocessing step).
Including the binary search for $t$, the time bound is $O(n\log^2 n\log\log n)$.

\subsection{Randomized Version 2}

Finally, we combine the ideas from both the deterministic and randomized implementations, to get our fastest randomized  algorithm for 3D halfspaces.

\subparagraph*{Solving Problem \COUNT.}
We may assume that all halfspaces are upper halfspaces.  We work in dual space, where $\hat{S}$ is now a multiset of points and $X$ is a set of planes.
In a query, we want to approximately count the number of points in $\hat{S}$ that are above a query plane in $X$.
By the sampling reduction from Section~\ref{sec:rand}, we may assume that the number of insertions to $\hat{S}$ is $O(t\log n)$.
By the preprocessing step from Section~\ref{sec:rand}, we may assume that all points in $\hat{S}$ have level at most $\frac nt$.

Compute a $\frac nt$-shallow $O(\frac 1t)$-cutting with $O(t)$ downward cells, along with its conflict lists.  For each point $p\in R$, locate the cell containing $p$.  All this can be done during a (one-time) preprocessing in $O(n\log n)$ time.

For each cell $\Delta$, we maintain $\hat{S}\cap\Delta$ in a semi-dynamic data structure for 3D approximate halfspace range counting.  As described in Section~\ref{sec:det},
we get $O(\log^2 n_\Delta)$ query and insertion time, where
$n_\Delta = |\hat{S}\cap\Delta|$.

In an insertion of a point $p$ to $\hat{S}$, we look up the cell $\Delta$ containing $p$ and insert the point to the approximate counting structure in $\Delta$.

In a query for a plane $h\in X$, we look up the cells $\Delta$
whose conflict lists contain $h$, answer approximate counting queries in these cells, and sum the answers.

We bound the total time for all insertions and queries.  For each cell $\Delta$,
the number of insertions in its approximate counting structure is $n_\Delta$ and
the number of queries is $O(\frac nt)$ (since each plane $h\in X$ is queried once).
The total time is
\[ O\left(\sum_\Delta ( \tfrac nt + n_\Delta) \log^2 n_\Delta\right). \]

\noindent Since there are $O(t)$ terms and $\sum_\Delta n_\Delta = O(t\log n)$, we have $n_\Delta=O(\log n)$ ``on average''; applying Jensen's inequality to the first term, we can
bound the sum by $O(n\log^2\log n + t\log^3 n)$.
Thus, $\Tcount = O(n\log^2\log n + t\log^3 n)$.

\subparagraph*{Conclusion.}
By~(\ref{eqn:runtime}), the complete randomized algorithm has running time
$O((\Treport + \Tcount)\log\frac nt + \Tnet) = O((n + n\log^2\log n + t\log^3 n)\log\frac nt + n\log n) = O(n\log n\log^2\log n + t\log^4 n)$.
If $t\le n/\log^3 n$, the first term dominates.
On the other hand, if $t > n/\log^3 n$, our earlier randomized algorithm has running time $O(n\log n\log\frac{n}{t})=O(n\log n\log\log n)$.
In any case, the time bound is at most $O(n\log n\log^2\log n)$.
Including the binary search for $t$, the time bound is $O(n\log n\log^3\log n)$.

\begin{theorem}
Given $O(n)$ points and $O(n)$ halfspaces in $\mathbb{R}^3$, we can find a subset of halfspaces covering all points, of size within $O(1)$ factor of the minimum,  in
$O(n\log n\log^3\log n)$ time by a randomized Monte-Carlo algorithm with error probability $O(n^{-c_0})$ for any constant~$c_0$.
\end{theorem}

\subparagraph*{Remark.} The number of the $\log\log n$ factors is improvable with still more effort, but we feel it is of minor significance.

\section{Weighted Set Cover}\label{sec:weighted}

In this final section, we consider the weighted set cover problem.
We define $\eps$-lightness and $\eps$-nets as before, ignoring the weights.  It is known that there exists an $\eps$-net of $S$ with total weight
$O(\frac{1}{\eps}\cdot \frac{w(S)}{|S|})$, for any set of 3D halfspaces or 2D disks (or objects in 2D with linear union complexity)~\cite{chan2012weighted}.  Here, the weight $w(S)$ of a set $S$ refers to the sum of the weights of the objects in $S$.


\subsection{MWU Algorithm in the Weighted Case}



\newcommand{\OPTT}{t}

Let $X$ be the set of input points and $S$ be the set of weighted input objects, where object $i$ has weight $w_i$, with $n=|X|+|S|$.
Let $\OPT$ be the weight of the minimum-weight set cover.
We assume that a value $\OPTT\in [\OPT,2\,\OPT]$ is given; this assumption can be removed by a binary search for $\OPTT$.

We may delete objects with weights $>\OPTT$.
We may automatically include all objects with weights $<\frac{1}{n}\OPTT$ in the solution, and delete them and all points covered by them, since the total weight of the solution increases by only $O(n\cdot \frac{1}{n}\OPTT)=O(\OPTT)$.
Thus, all remaining objects have weights in $[\frac{1}{n}\OPTT,\,\OPTT]$.
By rescaling, we may now assume that all objects have weights in $[1,n]$ and that $\OPTT=\Theta(n)$.

In the following,
for a multiset $\hat{S}$ where object $i$ has multiplicity $m_i$, the weight of the multiset is defined as $w(\hat{S})=\sum_i m_iw_i$.

We describe a simple variant of the basic MWU algorithm to solve the weighted set cover problem.  (A more general, randomized MWU algorithm for geometric set cover was given recently by Chekuri, Har-Peled, and Quanrud~\cite{fasterlp}, but our algorithm is simpler to describe and analyze.)
The key innovation is to replace doubling with multiplication by a factor $1+\frac{1}{w_i}$, where $w_i$ is the weight of the concerned object $i$.  (Note that multiplicities may now be non-integers.)

\begin{algorithm}[H]
\begin{algorithmic}[1]
\State Guess a value $\OPTT\in [\OPT, 2\,\OPT]$.
\State Define a multiset $\hat{S}$ where each object $i$ in $S$ initially has multiplicity $m_i=1$.
\Repeat  
    \State Find a point $p$ which is $\eps$-light in $\hat{S}$
    with $\eps=\frac{1}{2\OPTT}\cdot \frac{w(\hat{S})}{|\hat{S}|}$.
    \For {each object $i$ containing $p$}
    \Comment{call lines 5--6 a ``multiplicity-increasing step''}
        \State Multiply its multiplicity $m_i$ by $1+\frac{1}{w_i}$.
    \EndFor
\Until {all points are $\eps$-heavy in $\hat{S}$.}
\State Return an $\eps$-net of the multiset $\hat{S}$.
\end{algorithmic}
\end{algorithm}
Since at the end all points are $\eps$-heavy in $\hat{S}$, the returned subset is a valid set cover of $X$.  For halfspaces in 3D or disks in 2D, its weight is $O(\frac{1}{\eps}\cdot \frac{w(\hat{S})}{|\hat{S}|})=O(\OPT)$.

We now prove that the algorithm terminates in $O(\OPTT\log n)=O(n\log n)$ multiplicity-increasing steps.

%
In each multiplicity-increasing step, $w(\hat{S})$ increases by
\[\sum_{\text{object }i\text{ containing }p}m_i\cdot \tfrac{1}{w_i}\cdot w_i\ =\sum_{\text{object }i\text{ containing }p}m_i\ \leq\ \tfrac{w(\hat{S})}{2\OPTT},\]

\noindent
i.e., $w(\hat{S})$ increases by a factor of at most
$1+\frac{1}{2\OPTT}$. Initially, $w(\hat{S})\le n^2$.
Thus, after $z$ multiplicity-increasing steps, $w(\hat{S})\leq n^2(1+\frac{1}{2\OPTT})^z \le n^2 e^{z/(2\OPTT)}$.

On the other hand, consider the optimal set cover $T^*$.
Suppose that object $i$ has its multiplicity increased $z_i$ times. In each multiplicity-increasing step, at least one object in $T^*$ has its multiplicity increased. So, after $z$ multiplicity-increasing steps,
$\sum_{i\in T^*}z_i\geq z$ and $\sum_{i\in T^*}w_i\le\OPTT$.
In particular, $z_i/w_i\ge z/\OPTT$
for some $i\in T^*$.  Therefore,
$w(\hat{S})\geq (1+\frac{1}{w_i})^{z_i}w_i\geq (1+\frac{1}{w_i})^{z_i}\ge 2^{z_i/w_i}\ge 2^{z/\OPTT}$ (since $w_i\ge 1$).
We conclude that $2^{z/\OPTT}\le w(\hat{S})\le n^2 e^{z/(2\OPTT)}$,
implying that $z= O(\OPTT \log n)$.

Similar to Agarwal and Pan's first MWU algorithm, we can also divide the multiplicity-increasing steps into rounds, with each round performing up to $\OPTT$ multiplicity-increasing steps. Within each round, the total weight $w(\hat{S})$ increases by at most $(1+\frac{1}{2\OPTT})^{\OPTT}=O(1)$. Also if $|\hat{S}|$ increases by a constant factor, we immediately start a new round: because $|\hat{S}|\leq w(\hat{S})$ and $w(\hat{S})$ may be doubled at most $O(\log n)$ times, this case can happen at most $O(\log n)$ times. This ensures that if a point is checked to be $\eps$-heavy at any moment during a round, it will remain $\Omega(\eps)$-heavy at the end of the round. There are only $O(\log n)$ rounds.

Additional ideas are needed to speed up implementation (in particular, our modified MWU algorithm with multiplicity-readjustment steps does not work as well now).  First, we work with an approximation $\tilde{m}_i$ to the multiplicity $m_i$ of each object $i$.  By rounding, we may assume all weights $w_i$ are powers of 2.  In the original algorithm, $m_i=(1+\frac{1}{w_i})^{z_i}$, where $z_i$ is the number of points $p\in Z$ that are contained in object $i$, and $Z$ be the multiset consisting of all points $p$ that have undergone multiplicity-increasing steps so far.  Note that since the total multiplicity is $n^{O(1)}$, we have $z_i=O(w_i\log n)$.  Let $Y^{(w_i)}$ be a random sample of $Z$ where each point $p\in Z$ is included independently with probability $\frac{\log^2 n}{w_i}$ (if $w_i=O(\log^2 n)$, we can just set $Y^{(w_i)}=Z$).  Let $y_i$ be the number of points $p\in Y^{(w_i)}$ that are contained in object~$i$.  By the Chernoff bound, since $\frac{\log^2 n}{w_i} z_i=O(\log^3n)$, we have $|y_i - \frac{\log^2 n}{w_i} z_i| \le O(\log^2 n)$ with high probability.  By letting $\tilde{m}_i=(1+\frac{1}{w_i})^{y_i w_i/\log^2 n}$, it follows that $\tilde{m}_i$ and $m_i$ are within a factor of $O(1)$ of each other, with high probability, at all times, for all $i$.  Thus, our earlier analysis still holds when working with $\tilde{m}_i$ instead of $m_i$.  Since $z_i=O(w_i\log n)$, we have $y_i=O(\log^3 n)$ with high probability.  So, the total number of increments to all $y_i$ and updates to all $\tilde{m}_i$ is $O(n\log^3 n)$.  In lines 5--6, we flip a biased coin to decide whether $p$ should be placed in the sample $Y^{(2^j)}$ (with probability $\frac{\log^2 n}{2^j}$) for each $j$, and if so, we use halfspace range reporting in the dual to find all objects $i$ of weight $2^j$ containing $p$, and increment $y_i$ and update $\tilde{m}_i$.  Over all $O(n\log n)$ executions of lines 5--6 and all $O(\log n)$ indices $j$, the cost of these halfspace range reporting queries is $O(n\log n\cdot \log n\cdot \log n)$ plus the output size.  As the total output size for the queries is $O(n\log^3 n)$, the total cost is $O(n\log^3n)$.

We also need to redesign a data structure for lightness testing subject to multiplicity updates: For each $j$, we maintain a subset $S^{(j)}$ containing all objects $i$ with multiplicity at least $2^j$, in a data structure to support approximate depth (without multiplicity).  The depth of a point $p$ in $\hat{S}$ can be $O(1)$-approximated by $\sum_j 2^j\cdot \text{(depth of $p$ in $S^{(j)}$)}$.  Each subset $S^{(j)}$ undergoes insertion only, and the logarithmic method can be applied to each $S^{(j)}$.  Since $|\hat{S}|\le w(\hat{S})\le n^{O(1)}$, there are $O(\log n)$ values of~$j$.  This slows down lightness testing by a logarithmic factor, and so in the case of 3D halfspaces, the overall time bound is $O(n\log^4n\log\log n)$, excluding the $\eps$-net construction time.

\SOCG{
We can efficiently construct an $\eps$-net of the desired weight for 3D halfspaces in $O(n\log n)$ randomized time, by using the quasi-random sampling technique of Varadarajan~\cite{varadarajan2010weighted} and Chan et al.~\cite{chan2012weighted} in a more careful way.
Due to lack of space, we defer the description to the full paper.
We conclude:
}




\PAPER{

\subsection{Speeding up Quasi-Uniform Sampling}\label{app}

Finally, we show how to efficiently construct an $\eps$-net of the desired weight for 3D halfspaces.
We will take advantage of the fact that we need $\eps$-nets only in the discrete setting, with respect to a given set $X$ of $O(n)$ points.
Without loss of generality, assume that all halfspaces are upper halfspaces.

We begin by sketching (one interpretation of) the quasi-random sampling algorithm of Varadarajan~\cite{varadarajan2010weighted} and Chan et al.~\cite{chan2012weighted}:

\begin{algorithm}[H]
\begin{algorithmic}[1]
\State $k=\eps |S|$
\Repeat
\State Remove all points $p\in X$ with depth in $S$ less than $k$.
\State Move each point $p\in X$ downward so that its depth in $S$ is $\Theta(k)$.
\State  Pick a random sample $R\subseteq S$ of size $|S|/2+h$ for some appropriate choice of $h$.
\State Let $S'=S$. 
\Repeat
\State Find an object $i\in S'$ containing the fewest number of \emph{non-equivalent} points in $X$.
\If {object $i$ contains a \emph{bad} point}
add object $i$ to the output.
\EndIf
\State Remove object $i$ from $S'$. 
\Until {$S'$ is empty.}
\State Set $S\leftarrow R$ and $k\leftarrow k/2$.
\Until {$k$ is below a constant.}
\State Add $S$ to the output.
\end{algorithmic}
\end{algorithm}
In line~4, we use the property that the objects are upper halfspaces.
In line~8, two points $p$ and $q$ of $X$ are considered \emph{equivalent} iff the subset of objects from $S'$ containing $p$ is the same as the corresponding subset for $q$.
In line~9, a point $p$ is said to be \emph{bad} iff its depth in $S'$ is equal to $k$ and its depth in $R\cap S'$ is less than $k/2$.

With appropriate choices of parameters, in the case of 3D halfspaces, Chan et al.~\cite{chan2012weighted} showed that the output is an $\eps$-net, with the property that each object of $S$ is in the output with probability $O(\frac{1}{\eps |S|})$ (these events are not independent, so the output is only a  ``quasi-uniform'' sample).  This property immediately implies that the output has expected weight $O(\frac{1}{\eps}\cdot \frac{w(S)}{|S|})$. 
We will not redescribe the proof here, as our interest lies in the running time.

Consider one iteration of the outer repeat loop.
For the very first iteration, lines 3--4 can be done by answering $|X|$ halfspace range reporting queries in the dual (reporting up to $O(k)$ objects containing each query point $p\in X$), which takes $O(|S|\log |S| + |X|k)$ total time by using shallow cuttings (as described in Section~\ref{sec:det}).
As a result, we also obtain a list of the objects containing each point; these lists have total size $O(|X|k)$.  In each subsequent iteration, lines 3--4 take only $O(|X|k)$ time by scanning through these lists and selecting the $k$ lowest bounding planes per list.

For each point $p$, we maintain its depth in $S'$ and its depth in $R\cap S'$.  Whenever an object is removed from $S'$, we examine all points in the object, and if necessary, decrement these depth values; this takes $O(|X|k)$ total time (since there are $O(|X|k)$ object-point containment pairs).  Then line~9 can be done by scanning through all points in the object; again, this takes $O(|X|k)$ total time.

Line 8 requires more care, as we need to keep track of equivalence classes of points.  One way is to use hashing or fingerprinting~\cite{MotRag}:
for example, map each point $p$ to $(\sum_{\text{object $i\in S'$ containing $p$}} x^i)\bmod u$ for a random $x\in [u]$ and a fixed prime $u\in \Theta(n^c)$, where $c$ is a sufficiently large constant.  Then two points are equivalent iff they are hashed to the same value, with high probability.  When we remove an object $i$ from $S'$, we examine all points contained in the object, recompute the hash values of these points (which takes $O(1)$ time each, given a table containing $x^i\bmod u$), and whenever we find two equivalent points with the same hash values, we remove one of them.
This takes $O(|X|k)$ total time (since there are $O(|X|k)$ object-point containment pairs).
To implement line~8, for each object, we maintain a count of the number of points it contains.  Whenever we remove a point, we decrement the counts of objects containing it; again, this takes $O(|X|k)$ total time.  The minimum count can be maintained in $O(1)$ time per operation without a heap, since the only update operations are decrements (for example, we can place objects in buckets indexed by their counts, and move an object from one bucket to another whenever we decrement).

To summarize, the first iteration of the outer repeat loop takes $O(|S|\log |S| + |X|k)$ time, and each subsequent iteration takes $O(|X|k)$ time.
Since $k$ is halved in each iteration, the total time over all iterations is $O(|S|\log |S| + |X|k_0)$ where $k_0=\eps|S|$.

In our application, we need to compute an $\eps$-net of a \emph{multiset} $\hat{S}$.  Since the initial halfspace range reporting subproblem can be solved on the set $S$ without multiplicities,
the running time is still $O(|S|\log |S| + |X|k_0)$ but with $k_0=\eps|\hat{S}|$.

For $|S|,|X|=O(n)$, the time bound is $O(n\log n + nk_0)$, which is still too large.
To reduce the running time, we use one additional simple idea: take a random sample $R\subseteq \hat{S}$ of size $\frac{c}{\eps}\log n$ for a sufficiently large constant $c$.
Then $\mathbb{E}[w(R)] = O(\frac{w(\hat{S})}{\eps|\hat{S}|}\log n)$.
For a fixed point $p$ of depth $\ge \eps |\hat{S}|$ in $\hat{S}$, the depth of $p$ in $R$ is $\Omega(\eps|R|)=\Omega(\log n)$ with high probability, by the Chernoff bound.  We then compute a $\Theta(\eps)$-net of $R$, which gives us an $\eps$-net of $\hat{S}$, with
expected weight $O(\frac{1}{\eps} \cdot \frac{\mathbb{E}[w(R)]}{|R|}) = O(\frac{1}{\eps}\cdot \frac{w(\hat{S})}{|\hat{S}|})$.
The net for $R$ is easier to compute, since
$k_0$ is reduced to $\eps |R|=O(\log n)$.
The final running time for the $\eps$-net construction
is $O(|S|\log|S| + nk_0)=O(n\log n)$.

We can verify that the net's weight bound holds (and that all points of $X$ are covered), and if not, repeat the algorithm for $O(1)$ expected number of trials.

We conclude:
}

\begin{theorem}
Given $O(n)$ points and $O(n)$ weighted halfspaces in $\mathbb{R}^3$, we can find a subset of halfspaces covering all points, of total weight within $O(1)$ factor of the minimum, in
$O(n\log^4n\log\log n)$ expected time by a randomized Las Vegas algorithm.
\end{theorem}

\subparagraph*{Remark.}
A remaining open problem is to find efficient deterministic algorithms for the weighted problem.
Chan et al.~\cite{chan2012weighted} noted that the quasi-uniform sampling technique can be derandomized via the method of conditional probabilities, but the running time is high.


\bibliographystyle{plain}
\bibliography{references}

\end{document}